\documentclass{revtex4} 
\usepackage{amsthm,amssymb,amsmath} 
\usepackage{graphicx,epsfig} 
\usepackage{verbatim,enumerate}
\newcommand{\ket}[1]{\left|#1\right\rangle}
\newcommand{\bra}[1]{\left\langle #1\right|}

\newcommand{\Z}{\mathbb{Z}}
\newcommand{\re}{\mathbb{R}}
\newcommand{\cx}{\mathbb{C}} 
\newcommand{\De}{\Delta} 
\newcommand{\cA}{\mathcal{A}} 
\newcommand{\cH}{\mathcal{H}}

\newcommand{\conj}{\mathrm{conj}}
\newcommand{\cnot}{\mathrm{CNOT}} 

\newcommand{\modu}{\langle e^{\imath\theta} I \rangle}
\newcommand\defn[1]{\textsl{#1}}
\newtheorem{lemma}{Lemma}
\newtheorem{theorem}{Theorem} \newtheorem{proposition}{Proposition}
 
\begin{document}
\title{On the epistemic view of quantum states}
\author{Michael Skotiniotis}
\author{Aidan Roy} \author{Barry C. Sanders} \affiliation{Institute
for Quantum Information Science, University of Calgary,2500 University
Drive NW,Calgary AB, T2l 1N1, Canada}
\date{\today}
\begin{abstract}
We investigate the strengths and limitations of the Spekkens toy model,
which is a local hidden variable model that replicates many important
properties of quantum dynamics. First, we present a set of five axioms that
fully encapsulate Spekkens' toy model. We then test whether these axioms
can be extended to capture more quantum phenomena, by allowing operations on
epistemic as well as ontic states. We discover that the resulting group of
operations is isomorphic to the projective extended Clifford Group for two
qubits. This larger group of operations results in a physically unreasonable model; consequently, we claim that a relaxed definition of valid operations in Spekkens' toy model cannot produce an equivalence with the Clifford Group for
two qubits. However, the new operations do serve as tests for correlation in a two toy bit model, analogous to the well known Horodecki criterion for the separability of quantum states.
\end{abstract}
\maketitle
\section{\label{sec:intro}Introduction}

Spekkens introduced a toy theory that demonstrates how a
local hidden variable model with a classical information-based restriction can
capture a great deal of seemingly quantum phenomena, including
non-commutativity of measurement, remote steering and
teleportation~\cite{Spekkens04}. Spekkens' toy model (STM) builds upon other
information-based
models with similar aims \cite{Zeilinger, Fuchs, CBH,
Kirkpatrick}. Although by no means a proposed axiomatization of quantum theory,
STM aims to
strengthen the view that the quantum state is a statistical distribution over a
hidden variable space in which there exists a balance of knowledge and ignorance
about the true state of the system.

In this paper, we axiomatize STM, and test it by relaxing its axioms.  We claim
that STM can be formalized into five axioms describing valid states, allowable
transformations,
measurement outcomes, and composition of systems.  Arguing on empirical grounds,
we relax the axiom regarding valid
operations on toy bits to obtain larger groups of operations for one and two toy
bits.  We claim that these larger groups are isomorphic to the
projective extended Clifford Group for one and two
qubits respectively. However, these larger groups of operations contain elements
that do not necessarily
compose under the tensor product.  That is to say, there exist operations that
do not take valid states to valid states when composed under the tensor product,
as one would demand of a physical model.  These operations are analogous to
positive maps in quantum theory.  Just as positive (but not completely positive)
maps can be used to test whether a quantum state is entangled or
not~\cite{Horodecki's}, validity-preserving (but not completely
validity-preserving) maps can be used to test for correlations in the two toy bit STM. Finally, we claim that relaxing the transformations of STM to an epistemic perspective gives rise to physically unreasonable alternatives, and as such, no equivalence with the extended Clifford Group for two qubits can be established by relaxing STM's operations.

The outline of the paper is as follows. In Section \ref{sec:STM}, we present STM
as a series of axioms and compare them to the axioms of quantum theory.  We
provide a
brief review of the original model for an elementary toy system (a toy bit) and
for two toy bits and provide a number of
different ways of representing one and two toy bits. In
Section~\ref{sec:relax} we propose a relaxation of the
criterion for valid operations on elementary systems, identify the resulting
groups of operations, and analyze both their mathematical and physical
properties.
We conclude with a discussion of our results
in Section~\ref{sec:conclusion}.

\section{\label{sec:STM}The Spekkens toy model and quantum theory}

In this section we present STM in its axiomatic basis and state the axioms of
quantum mechanics for comparison. Using the axioms of STM we develop several
ways of representing toy bits including a
vector space, a tetrahedron, and a toy analogue of the Bloch sphere. We also
develop two ways of representing two toy bits: a product
space and a four-dimensional cube.  We show how states, operations, and tensor
products stem from the axioms of STM, and we draw parallels to the equivalent
axioms and concepts in quantum theory.

STM is based on a simple classical principle called
the \emph{knowledge balance principle}:\begin{quote}If one has maximal
knowledge, then for every system, at every time,
the amount of knowledge one possesses about the ontic state of the
system at that time must equal the amount of knowledge one lacks.\end{quote}

Spekkens realizes the knowledge balance principle using canonical sets of yes/no
questions, which
are minimal sets of questions that completely determine the actual state of a
system. For any given
system, at most half of a canonical set of questions can be answered. The state
a system is actually in is called an ontic state, whereas the state of
knowledge is called an epistemic state. 

STM can be succinctly summarized using the following axioms:
\begin{description}
\item[STM~0:] All systems obey the knowledge balance principle. 
\item[STM~1:] A single toy bit is described by a single hidden variable that can
be in 1 of 4 possible states, the \defn{ontic} states. The knowledge balance
principle insists that the hidden variable is known to be in a subset of 2 or 4
of the ontic states---that subset is the \defn{epistemic} state of the system.
\item[STM~2:] A \defn{valid reversible operation} is a permutation of ontic
states of the system that also permutes the epistemic states amongst themselves.
\item[STM~3:] A \defn{reproducible measurement} is a partition of the ontic
states into
a set of disjoint epistemic states, with the outcome of a measurement being a
specific epistemic state. The probability of a particular outcome is
proportional to the number of ontic states that outcome has in common with the
current epistemic state.  Immediately after the process of measurement, the
epistemic state of the system is updated to the outcome of the measurement.
\item[STM~4:]  Elementary systems compose under the tensor product giving rise
to composite systems; the knowledge balance principle applies to the composite
system as well as to the parts.
\end{description}

To help make the comparison with quantum theory, the corresponding axioms of
quantum mechanics
are given below \cite{MikeIke}.
\begin{description}
\item[QM~1:] Any isolated physical system corresponds to a complex vector space
with an inner product, a \defn{Hilbert space}.  A system is completely described
by a ray in Hilbert space.
\item[QM~2:] Evolution of a closed system is described by a unitary
transformation through the Schr\"{o}dinger equation
\begin{equation}\label{schrodinger}
\hat{H}\ket{\psi}=\imath\hslash\frac{\partial\ket{\psi}}{\partial t}
\end{equation}
whereas $\hat{H}$ is a Hermitian operator.
\item[QM~3:]  Measurement is described by a collection, $\{M_m\}$, of
measurement operators.  These are operators acting on the state space of the
system being
measured.  The index $m$ refers to the measurement outcomes that may occur in
the experiment.  If the state of the quantum system is $\ket{\psi}$ immediately
before the measurement then the probability that result $m$ occurs is given by
 \begin{equation}\label{measurement}
 p(m)=\bra{\psi}\hat{M}_m^{\dagger}\hat{M}_m\ket{\psi},
 \end{equation}
and the state after measurement is given by 
 \begin{equation}
 \frac{\hat{M}_m\ket{\psi}}{\sqrt{\bra{\psi}\hat{M}_m^{\dagger}\hat{M}_m\ket{
\psi}}}.
 \end{equation}
Measurement operators satisfy $\sum_m \hat{M}_m^{\dagger}\hat{M}_m=I$.
\item[QM~4:]  The state space of a composite system is the tensor product of the
 state space of the component systems.
\end{description}

The simplest system that can exist is a single toy bit system: there are two
yes/no questions in a canonical set, yielding
four ontic states, which we label $o_1$, $o_2$, $o_3$, and $o_4$. A pair of
ontic states forms the answer to one of the two questions in a canonical set.
The knowledge balance principle restricts us to
knowing the answer to at most one of two questions, resulting in a \defn{pure
epistemic
state}.  The six pure states are shown pictorially in
Fig.~\ref{fig:pureepistemic}. (In Spekkens' original notation, the state
$e_{ij}$ was denoted $i \vee j$.)
\begin{figure}[htb]
\begin{center}
\includegraphics[keepaspectratio,width=75mm]{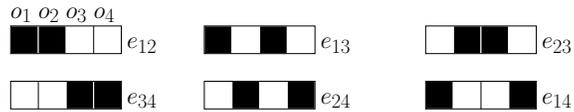}
\caption{\small The six pure epistemic states of the single toy bit model.}
\label{fig:pureepistemic}
\end{center}
\end{figure}

By way of example, the questions ``Is the ontic state in $\{o_1,o_2\}$?" and
``Is the ontic state in $\{o_1,o_3\}$?" form one particular canonical set. The
epistemic state $e_{12} = o_1 + o_2$ corresponds to the situation in which the
first question can be answered, and it is in the affirmative. The model also
includes a single \defn{mixed epistemic state}, namely $e_{1234} = o_1 + o_2 +
o_3 + o_4$, corresponding to knowing absolutely nothing about the system.

At this point we introduce the linear represention for the toy model which will
be
convenient for describing operations later. Let $\{o_1,o_2,o_3,o_4\}$ be a basis
for a real vector space, and express the epistemic states in that basis. Each
pure epistemic state is then a vector with exactly two $1$'s and two $0$'s; for
example,
\[
e_{12} =\left(\begin{array}{c} 
    1 \\1 \\ 0 \\ 0
  \end{array}\right).
\]
Note that epistemic states that are disjoint (that is, have no ontic states in
common) are orthogonal as vectors in $\re^4$.

Now that states in the toy model are defined, we turn our attention to
transformations between states.  STM~2 states that valid operations are
permutations of ontic states. The group of permutations
of four objects is denoted $S_4$, and permutations are usually summarized using
cyclic notation (see \cite[p.~7]{Spekkens04} for details). By way of example,
the permutation $(123)(4)$ maps $o_1$ to $o_2$, $o_2$ to $o_3$, $o_3$ to $o_1$,
and $o_4$ to $o_4$. In terms of epistemic states, $(123)(4)$ maps $e_{12}$ to
$e_{23}$. In the linear representation, each transformation in $S_4$ is a $4
\times 4$
permutation matrix that acts on the left of the epistemic state vectors. For
example,
\begin{equation}
(123)(4)=\left(\begin{array}{cccc} 
    0 & 0 & 1 & 0 \\
    1 & 0 & 0 & 0 \\
    0 & 1 & 0 & 0 \\
    0 & 0 & 0 & 1 
\end{array}\right).
\end{equation}
We call this the \defn{regular representation} of $S_4$, and we will call this
description of STM the \defn{linear model}.

Since the group of operations on a single toy bit is such a well-studied group,
there are other classical systems of states and transformations that may be
readily identified with the single toy bit. One such system uses a regular
tetrahedron. In this geometric representation, the vertices of the tetrahedron
represent the ontic states of the system, whereas pure epistemic states are
represented by edges (see Fig.~\ref{fig:tetrahedron}). 
\begin{figure}[htb]
\begin{center}
\includegraphics[keepaspectratio,width=35mm]{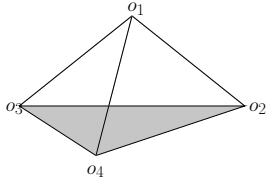}
\caption{\small The regular tetrahedron representation of a toy bit.}
\label{fig:tetrahedron}
\end{center}
\end{figure}
The action of a transformation in $S_4$, then, is a symmetry operation on the
tetrahedron. For example, the transformation $(123)(4)$ permutes vertices $o_1$,
$o_2$, and $o_3$ of the tetrahedron by rotating counter-clockwise by
$2\pi/3$ about the axes which passes through the center of
the tetrahedron and vertex $o_4$.  Since $S_4$ is the entire group of
permutations of $\{o_1,o_2,o_3,o_4\}$, it is
also the complete group of symmetry operations for the regular tetrahedron.
Notice that $A_4$, the alternating group (or group of even permutations),
corresponds to the group of rotations, whereas odd permutations correspond to
reflections and roto-reflections.

As pointed out by Spekkens, another way of viewing the single toy bit is using a
toy analogue of the Bloch sphere. In the toy Bloch sphere, epistemic states are
identified with particular quantum states on the traditional Bloch sphere and
are embedded in $S^2$ accordingly. In particular, $e_{13}$, $e_{23}$, and
$e_{12}$ are identified with $\ket{+}$, $\ket{\imath}$, and $\ket{0}$ and are
embedded on the positive $x$, $y$, and $z$ axes respectively (see
Fig.~\ref{fig:Bloch}). 
\begin{figure}[htb]
\begin{center}
\includegraphics[keepaspectratio,width=60mm]{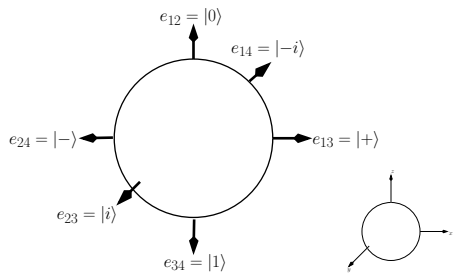}
\caption{\small The Bloch sphere, with both toy and quantum labels.}
\label{fig:Bloch}
\end{center}
\end{figure}
States that are orthogonal in the linear model are
embedded as antipodal points on the toy Bloch sphere, just as orthogonal quantum
states are embedded antipodally on the quantum Bloch sphere. Distance on the toy
Bloch sphere corresponds to overlap between states: two epistemic states have an
angle of $\pi/2$ between them if and only if they have exactly one ontic state
in common.

On the quantum Bloch sphere, single qubit transformations are represented by
rotations in the group $SO(3)$, and they may be characterized using Euler
rotations. More precisely, if $R_x(\theta)$ denotes a rotation about the
$x$-axis by $\theta$, then any $T \in SO(3)$ may be written in the form 
\begin{equation}
T = R_x(\theta)R_z(\phi)R_x(\psi), \qquad  0 \leq \theta \leq \pi, \;\; -\pi <
\phi,\psi \leq \pi.
\end{equation}
For example, the rotation by $2\pi/3$ about the $x+y+z$ axis may be written as
$R_{x+y+z}(\frac{2\pi}{3}) =
R_x(\pi)R_z\left(\frac{-\pi}{2}\right)R_x\left(\frac{-\pi}{2}\right)$ (see
Fig.~\ref{fig:euler}).
\begin{figure}[htb]
\includegraphics[keepaspectratio,width=150mm]{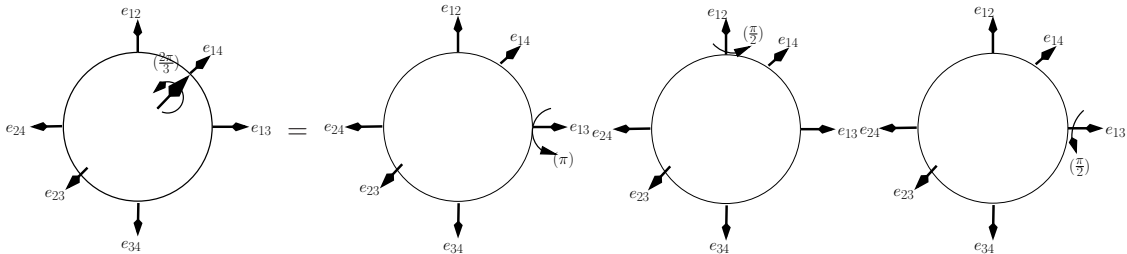}
\caption{\small The element $R_{x+y+z}(2\pi/3)$ expressed as a series of Euler
rotations.}
\label{fig:euler}
\end{figure}

On the toy Bloch sphere, in contrast, transformations are elements of $O(3)$,
not all of which are rotations. For example, the permutation $(12)(3)(4)$ is not
a rotation of the toy Bloch sphere but a reflection through the plane
perpendicular to the $x-y$ axis (see Fig.~\ref{fig:conj}). Thus, there are
operations in the single toy bit model that have no quantum analogue. (We will
see shortly that such toy operations correspond to \defn{anti-unitary} quantum
operations.)
\begin{figure}[htb]
\begin{center}
\includegraphics[keepaspectratio,width=35mm]{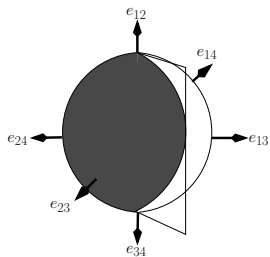}
\caption{\small The element $(12)(3)(4)$ acts as a reflection on
the toy sphere.}
\label{fig:conj}
\end{center}
\end{figure}

The toy operations that do correspond to rotations on the Bloch sphere are
precisely the operations in $A_4$, the group of even permutations. In terms of
the linear model, these are the transformations of $S_4$ with determinant $1$.
Toy operations not in $A_4$ may be expressed as a rotation composed with a
single reflection. When $T$ is a rotation on the toy Bloch sphere, its Euler
rotations $R_x(\theta)R_z(\phi)R_x(-\psi)$ satisfy $\theta \in \{0,\pi/2,\pi\}$
and $\phi,\psi \in \{-\pi/2,0,\pi/2,\pi\}$. For example, the permutation
$(123)(4)$ corresponds to the rotation $R_{x+y+z}(2\pi/3)$ seen in
Fig.~\ref{fig:euler}. 

STM~3 addresses the problem of measurement in the toy theory.
For a single toy bit, a measurement is any one question from a canonical set;
thus there are a
total of six measurements that may be performed. After a
measurement is performed and a result is obtained, the observer has acquired new
information about the system and updates his state of knowledge to the result
of the measurement.
This ensures that a repeat of the question produces the same outcome. Note that
the outcome of a measurement is governed by the ontic state of the system and
not the measurement itself. 

The question ``Is the ontic state in $\{o_m,o_n\}$?'' can be represented by a
vector $r_{mn} = o_m + o_n$. The probability of getting ``yes" as the outcome is
then
\begin{equation}\label{toymeasurement}
p_{mn}= \frac{r_{mn}^T e_{ij}}{2},
\end{equation}
where $e_{ij}$ is the current epistemic state of the system. After this outcome,
the epistemic state is updated to be $e_{mn}$. The vectors $r_{mn}$ and
probabilities $p_{mn}$ are analogous to the measurement operators and outcome
probabilities in QM~3.

STM~4 concerns the composition of one or more toy bits.
For the case of two toy bits there are four questions in a canonical
set, two per bit, giving rise to 16 ontic states, which we denote
$o_{ij},\ i,j=1\ldots 4$.  In the linear model this is simply the
tensor product of the $4$-dimensional vector space with itself, and the ontic
state $o_{ij}$ is understood to be
$o_i\otimes o_j$. The types of epistemic states arising in this case
are of three types; maximal, non maximal, and zero knowledge,
corresponding to knowing the answers to two, one, or zero questions
respectively.  It suffices for our purposes to consider only states of maximal
knowledge (pure states).  These, in Spekkens'
representation, are of two types (see Fig.~\ref{fig:twostates}),
\begin{figure}[htb]
\begin{center}
\includegraphics[keepaspectratio, width=55mm]{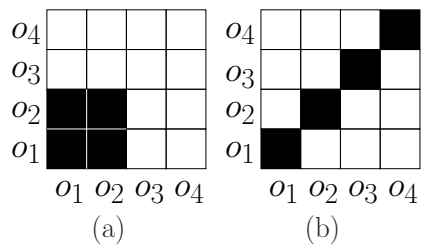}
\caption{\small (a) Uncorrelated and (b) correlated states in the toy
model.}
\label{fig:twostates}
\end{center}
\end{figure}
called uncorrelated and correlated states respectively. An \defn{uncorrelated
state} is the tensor product of two pure single toy bit states. If each of the
single toy bits satisfy the knowledge balance principle, then
their composition will also satisfy the knowledge balance principle for the
composite system.  A \defn{correlated state} is one in which nothing is known
about the ontic state of each elementary system, but
everything is known about the classical correlations between the ontic states of
the two elementary toy systems. If the two single bit systems in
Fig.~\ref{fig:twostates}(b) are labelled A and B, then nothing is known about
the true state of either A or B, but we know that if A is in the state $o_i$,
then B is also in the state $o_i$. 

According to STM~2, operations on two toy bits are permutations of ontic states
that map epistemic states to epistemic states.  These permutations are of two
types: tensor products of permutations on the individual systems, and 
indecomposable permutations (see Fig.~\ref{fig:perm}).
\begin{figure}[htb]
\begin{center}
\includegraphics[keepaspectratio,width=55mm]{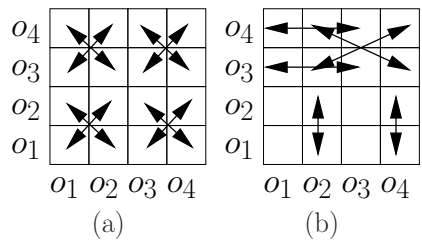}
\caption{\small Operations on two toy bits: (a) a tensor product operation and
(b) an indecomposable permutation.}
\label{fig:perm}
\end{center}
\end{figure}
Moreover, STM~4 suggests that if an operation is valid on a given system, then
it should still be valid when an ancilla is added to that system. That is, if
$T$ is a valid operation on a single toy bit, then $T \otimes I$ ought to be
valid
on two toy bits.  It follows that valid operations should compose under the
tensor product. 

Finally, STM~3 implies that a measurement of the two toy bit space is a
partition of ontic states into disjoint epistemic states: each epistemic state
consists of $4$ or $8$ ontic states. There are in total $105$ partitions of the
two toy bit space into epistemic states of size $4$.

In the linear model, epistemic states, operations, and measurements extrapolate
in the manner anticipated. A pure epistemic state is a $\{0,1\}$-vector
of length $16$ containing exactly $4$ ones, whereas an operation is a
$16\times 16$ permutation matrix.  The group of operations can be
computationally verified to be of order $11520$. Measurement is a row vector
$r_{{o_{ijkl}}}\in \{0,1\}^{16}$ with the state after measurement updated
according to the outcome obtained.  In the linear model STM~4 is understood as
the composition of valid states and operations under the tensor product.  

Finally, a two toy bit system can be geometrically realized by the
four-dimensional cube (see Fig.~\ref{tesseract}).  This is a new
representation for the two toy bit system that in some ways
generalizes Spekkens' tetrahedral description of the single toy bit.
By mapping the ontic states $o_1\ldots o_4$ of an elementary system to
the vertices $(x,y),\ x,y\in\{-1,1\}$ of a square, the four-dimensional cube is
the result of the tensor product of two elementary systems.  Every
epistemic state is an affine plane containing four vertices, and the
group of permutations of two toy bits is a subgroup of $B_4[3,3,4]$,
the symmetry group of the four-dimensional cube (for more details, see
\cite{Coxeter}).
\begin{figure}[htb] 
\begin{center}
\includegraphics{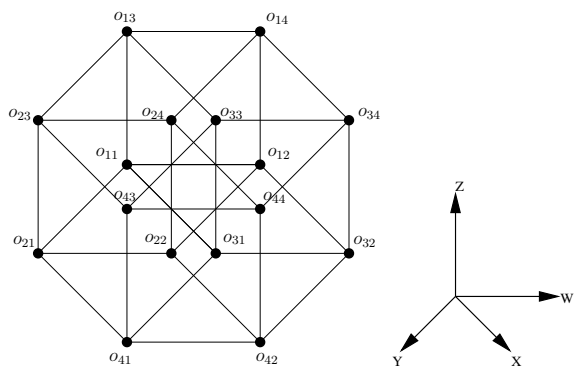}
\caption{\small The four-dimensional hypercube representation for the space
of two toy bits.}
\label{tesseract}
\end{center}
\end{figure}

In this section we reviewed STM, identifying its axioms and drawing a
correspondence with the axioms of quantum theory. In the next section, we 
investigate a relaxation of STM~2.

\section{\label{sec:relax}Relaxing the Spekkens toy model.}

In this section we relax STM~2, the axiom describing valid reversible
operations. We obtain a new group of operations which contains a subgroup
isomorphic to the projective Clifford Group for two qubits, a characteristic of
quantum theory not captured by STM.  However,
the operations in these new group fail to compose under the tensor product,
rendering the relaxation of STM~2 physically unreasonable. Nevertheless, we
claim that operations that fail to compose under the tensor product can be used
as tests for correlations in STM.

Recall that STM~2 describes valid operations on toy states.  In particular,
STM~2 requires that valid operations act on the ontic states in a reversible
manner
(ontic determinism).  Now consider an empiricist living in a universe
governed by the axioms of STM---a toy universe.  Such an empiricist has access
only to epistemic states.  As a result an empiricist sees determinism only
at the epistemic scale (epistemic determinism); the knowledge balance
principle
forbids exact knowledge of the ontic state of the system. For an empiricist,
ontic determinism is
too strict a condition. We thus propose the following amendment.
\begin{description}
\item[STM~2$'$:] A \defn{valid reversible operation} is a linear transformation
that permutes the epistemic states of the system.
\end{description}
.
The requirement that transformations be linear implies that as
$e_{1234} = e_{12}+e_{34}$, then $T(e_{1234}) = T(e_{12})+T(e_{34})$ for any
valid $T$: in other words, mixtures of epistemic states are transformed into
other mixtures.  It follows that pairs of disjoint epistemic states are mapped
to other pairs of disjoint states, and the amount of overlap between epistemic
states is preserved. This linearity condition is essential if the toy theory is
to emulate
significant aspects of quantum theory.  Investigations into a non-linear theory
of quantum mechanics~\cite{IBB, Weinberg1, Weinberg2} have been experimentally
tested and found to be ``measurably not different from the linear
formalism''~\cite{Itano}. Furthermore it was shown by Peres that a non-linear
quantum
mechanical theory would violate the second law of thermodynamics~\cite{Peres}.

We let $TG(1)$ denote the group of operations obtained by replacing STM~2 with
STM~2$'$. In terms of
the linear model, an operation is in $TG(1)$ if it can be represented as a $4
\times 4$ orthogonal matrix that maps epistemic states to epistemic states.
This includes all the operations in $S_4$, but it also includes operations such
as
\begin{equation} \label{newoperations}
\widetilde{\sqrt{Z}}=\frac{1}{2}\left(\begin{array}{c c c c} 
    1 & 1 & -1 & 1\\
    1 & 1 & 1 & -1\\
    1 & -1 & 1 & 1\\
    -1 & 1 & 1 & 1\end{array}\right), \quad
\widetilde{H}=\frac{1}{2}\left(\begin{array}{c c c c}
    1 & 1 & 1 & -1\\
    1 & -1 & 1 & 1\\
    1 & 1 & -1 & 1\\
    -1& 1 & 1 & 1
\end{array}\right).
\end{equation}

On the toy Bloch sphere, $TG(1)$ is the subgroup of operations in $O(3)$ that
preserve the set of six pure epistemic states.  On the toy Bloch sphere,
Eq.~\eqref{newoperations}, are the Euler rotations
\begin{equation}\label{oprot}
\widetilde{\sqrt{Z}}=R_z\left(-\frac{\pi}{2}\right),\quad
\widetilde{H}=R_x\left(\frac{\pi}{2}\right)R_z\left(\frac{\pi}{2}
\right)R_x\left(\frac{\pi}{2}
\right),
\end{equation}
respectively. We have called these operations $\widetilde{\sqrt{Z}}$ and
$\widetilde{H}$ because their action on the toy Bloch sphere resembles the
quantum operations $\sqrt{Z}$ and $H$ respectively.

The order of $TG(1)$ is 48, as the next lemma shows.
\begin{lemma}
$TG(1)$ is the set of all permutations of $\{e_{13},e_{24},e_{23},e_{14},e_{12},
e_{34}\}$ such that pairs of antipodal states are mapped to pairs of antipodal
states.
\label{lem:tg1}
\end{lemma}
\begin{proof}
Since $TG(1)$ contains $S_4$ as a proper subgroup, $TG(1)$ has order at least
$48$. Moreover, every element of $TG(1)$ is a permutation of epistemic states
mapping pairs of antipodal points to pairs of antipodal points. We prove the
lemma by counting those permutations; as only $48$ such operations exist,
they must all be in $TG(1)$.

There are three pairs of antipodal states on the toy sphere, namely
$\{e_{13},e_{24}\}$, $\{e_{23},e_{14}\}$, and $\{e_{12}, e_{34}\}$. Therefore a
map that preserves pairs of antipodal points must permute these three pairs:
there are $3! = 6$ such permutations. Once a pair is chosen, there are two ways
to permute the states within a pair. Therefore, there are a total of $3!\cdot
2^3 = 48$ distinct permutations that map pairs of antipodal states to pairs of
antipodal states.
\end{proof}

By the argument in Lemma \ref{lem:tg1}, $TG(1)$ may be formally identified with
the semidirect product
$(\Z_2)^3\rtimes S_3$, where $g\in S_3$ acts on $\Z_2^3$ by
\begin{equation}\label{Operation}
  g:(x_1,x_2,x_3)\mapsto (x_{g(1)},x_{g(2)},x_{g(3)}),\quad (x_1,x_2,x_3)\in
  \Z_2^3.
\end{equation}
An element of $S_3$ permutes the three pairs of antipodal states, whereas an
element of $\Z_2^3$ determines whether or not to permute the states within each
antipodal pair. The following result explains how Spekkens' original group of
operations fits into $TG(1)$.
\begin{lemma}
$S_4$ is the subgroup of $\Z_2^3\rtimes S_3$ consisting of elements
$((x,y,z),g)$ such that $(x,y,z) \in \Z_2^3$ has Hamming weight of
zero or two.
\end{lemma}
\begin{proof}
Label the antipodal pairs $\{e_{13},e_{24}\}$, $\{e_{23},e_{14}\}$, and
$\{e_{12}, e_{34}\}$ with their Bloch sphere axes of $x$, $y$, and $z$. Now
$S_4$ is generated by the elements $(12)(3)(4)$, $(23)(1)(4)$, and $(34)(1)(2)$,
and by considering the action on the Bloch sphere, we see that these elements
correspond to $((0,0,0),(z)(xy))$, $((0,0,0),(zx)(y))$ and $((1,1,0),(z)(xy))$
in $\Z_2^3\rtimes S_3$ respectively.  Note that $((0,0,0),(z)(xy))$ and
$((0,0,0),(zx)(y))$ generate all elements of the form $((0,0,0),g)$ with $g \in
S_3$, so adding $((1,1,0),(z)(xy))$ generates all elements of the form
$((x,y,z),g)$ whereas $(x,y,z)$ has Hamming weight zero or two.
\end{proof}

$TG(1)$ exhibits a relationship with the operations in quantum mechanics
acting on a single qubit restricted to the six states shown in
Fig.~\ref{fig:Bloch}. To describe the connection, we must first describe the
extended
Clifford Group. 

Recall that the \defn{Pauli Group} for a single qubit, denoted $\mathcal{P}(1)$,
is the group of matrices generated by $X=(\begin{smallmatrix} 0 & 1 \\ 1 &
0\end{smallmatrix})$ and $Z=(\begin{smallmatrix} 1 & 0\\ 0 &
-1\end{smallmatrix})$. The \defn{Clifford Group}, denoted $\mathcal{C}(1)$, is
the normalizer of the Pauli Group in $U(2)$, and is generated by the matrices
(see \cite{Marten})
\begin{equation}
H=\frac{1}{\sqrt{2}}\left(\begin{matrix}
    1 & 1 \\
    1 & -1\end{matrix}\right), \; \sqrt{Z}=\left(\begin{matrix}
    1 & 0\\
    0 & i\end{matrix}\right), \; \left\{ e^{\imath\theta}I \mid 0 \leq \theta <
2\pi\right\}.
\end{equation}
Since $U$ and $e^{\imath\theta}U$ are equivalent as quantum operations, we focus
on the projective group of Clifford operations, namely $\mathcal{C}(1)/U(1)
\cong \mathcal{C}(1)/\modu$. This is a finite group of $24$ elements. For our
purposes, the significance of the Clifford
Group is that it is the largest group in $U(2)$ that acts invariantly on the
set of the six quantum states
$\{\ket{0},\ket{1},\ket{+},\ket{-},\ket{\imath},\ket{-\imath}\} \subset \cx^2$
(with $\ket{\psi}$ and $e^{\imath\theta}\ket{\psi}$ considered equivalent).

An \defn{anti-linear} map on $\cx^2$ is a transformation $T$ that satisfies the
following condition for all $u,v \in \cx^2$ and $\alpha \in \cx$:
\begin{equation}
T(\alpha u + v) = \bar{\alpha}T(u) + T(v).
\end{equation}
Every anti-linear map may be written as a linear map composed with the
complex conjugation operation, namely
\begin{equation}
\conj: \alpha \ket{0}+\beta \ket{1} \mapsto
\bar{\alpha}\ket{0}+\bar{\beta}\ket{1}.
\end{equation}
An \defn{anti-unitary} map is an anti-linear map that may be written as a
unitary map composed with conjugation. The unitary maps $U(2)$ and their
anti-unitary counterparts together form a group, which we denote $EU(2)$.
Finally, the \defn{extended Clifford Group} $\mathcal{EC}(1)$ is the normalizer
of the Pauli Group in $EU(2)$. Working projectively, $\mathcal{EC}(1)/U(1)$ is a
finite group of $48$ elements, generated by $\sqrt{Z}\modu$, $H\modu$, and
$\conj\modu$. For more details about the extended Clifford Group, see for
example~\cite{Appleby}.

The following proposition demonstrates the relationship between $TG(1)$ and
$\mathcal{EC}(1)/U(1)$.
\begin{proposition} 
The toy group $TG(1)$ is isomorphic to the projective extended Clifford Group
$\mathcal{EC}(1)/U(1)$.
\end{proposition}
\begin{proof}
By Lemma \ref{lem:tg1}, $TG(1)$ consists of all possible ways of permuting
$\{e_{13},e_{24},e_{23},e_{14},e_{12}, e_{34}\}$ such that antipodal points are
mapped to antipodal points. Now consider the quantum analogues of these states,
namely $\ket{+},
\ket{-}, \ket{\imath}, \ket{-\imath}, \ket{0}$, and $\ket{1}$ respectively. For
each $T\modu$ in $\mathcal{EC}(1)/U(1)$, $T$ is a normalizer of the Pauli Group,
so $T\modu$ acts invariantly on the six quantum states as a set. Since $T$ is
also unitary or anti-unitary, it preserves distance on the Bloch sphere and
therefore maps antipodal points to antipodal points. By the argument in Lemma
\ref{lem:tg1}, there are only $48$ such operations, and it is easy to verify
that no two elements of $\mathcal{EC}(1)/U(1)$ act identically. It follows that
$\mathcal{EC}(1)/U(1)$ and $TG(1)$ are isomorphic, as both are the the group of
operations on six points of the Bloch sphere that map pairs of antipodal points
to pairs of antipodal points.
\end{proof}

We now look at the composition of two elementary systems. In the linear model of
two toy bits, every valid operation is an orthogonal matrix.  As
STM~2$'$ requires that valid operations map epistemic states to
epistemic states reversibly, it can be shown that operations such as
$I\otimes\widetilde{H}P$, with $P \in S_4$, fail to map correlated states
to valid epistemic states and therefore are not valid
operations. On the other hand, operations such as
$P\widetilde{H}\otimes Q\widetilde{H}$, with $P,Q \in S_4$, are valid under
STM~2$'$.  Let $TG(2)$ denote the group of valid operations for two toy bits.
The order of $TG(2)$ can be verified computationally to be $23040$, and
Spekkens' group of operations is a subgroup of $TG(2)$. 

We discover that $TG(2)$ is very simply
related to the \defn{extended Clifford Group} for two qubits,
$\mathcal{EC}(2)$.  Let $\mathcal{P}(2)$ be the Pauli Group for two
qubits; then the extended Clifford Group for two qubits,
$\mathcal{EC}(2)$, is the group of all unitary and anti-unitary
operators $U$ such that
\begin{equation}
U\mathcal{P}(2)U^{\dagger} =\mathcal{P}(2).
\end{equation} 
It is generated by 
\begin{equation}
\sqrt{Z}\otimes I, \; I\otimes \sqrt{Z}, \; H \otimes I, \; I \otimes H, \;
\cnot =
\left(\begin{matrix}
    1 & 0 & 0 & 0 \\
    0 & 1 & 0 & 0 \\
    0 & 0 & 0 & 1 \\
    0 & 0 & 1 & 0 \\
\end{matrix}\right),
\end{equation}
the conjugation operation, and unitary multiples of the identity matrix. Working
projectively, it can be
shown that $\mathcal{EC}(2)/U(1)$ is a group of order $23040$
(see \cite{Appleby}). The two-qubit Clifford Group $\mathcal{C}(2)$ is a
subgroup of $\mathcal{EC}(2)$, and $\mathcal{C}(2)$ is the largest group in
$U(4)$ that acts invariantly on a set of sixty states; this is the same size as
the set of epistemic states for two toy bits.  The following isomorphism was
verified using the computation program GAP \cite{GAP}.

\begin{proposition}
  $TG(2)$ is isomorphic to $\mathcal{EC}(2)/U(1)$, the
  two qubit extended Clifford Group modulo phases.
\end{proposition}

\begin{figure}
\begin{center}
\includegraphics[keepaspectratio,width=120mm]{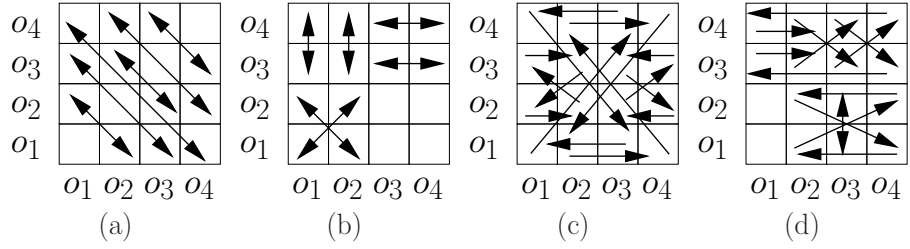}
\caption{\small (a) $\widetilde{SWAP}$, (b) $\widetilde{P_1}$, (c)
$\widetilde{P_2}$ and (d) $\widetilde{P_3}$: four operations on two toy bits.}
\label{fig:tswap}
\end{center}
\end{figure}

We give one such isomorphism explicitly. Let $\widetilde{SWAP}$ denote the toy
operation that swaps rows and columns of ontic states, and let $\widetilde{P_1}$ and $\widetilde{P_2}$ be as shown in Fig.~\ref{fig:tswap}. 
For convenience, we use the generating set $\{\conj, \cnot, H \otimes
I, H \otimes H, \sqrt{Z} \otimes \sqrt{Z} \}$ for $\mathcal{EC}(2)$.
Then the following map, extended to the entire group, is an
isomorphism from $\mathcal{EC}(2)/ U(1)$ to
$TG(2)$:
\begin{align*}
  \conj \langle e^{\imath\theta}I\rangle & \;\mapsto\; \frac{1}{4}
  \left(\begin{matrix}
      1 & 1 & -1 & 1 \\
      1 & 1 & 1 & -1 \\
      -1 & 1 & 1 & 1 \\
      1 & -1 & 1 & 1 \\
    \end{matrix}\right)^{\otimes 2}, \\
  \cnot \langle e^{\imath\theta} I\rangle & \;\mapsto\; \widetilde{SWAP} \cdot
\frac{1}{4}
\left(\begin{matrix}
      -1 & 1 & 1 & 1 \\
      1 & 1 & -1 & 1 \\
      1 & -1 & 1 & 1 \\
      1 & 1 & 1 & -1 \\
  \end{matrix}\right)^{\otimes 2},\\
(H \otimes I) \langle e^{\imath\theta} I\rangle & \;\mapsto\; \frac{1}{4}
\left(\begin{matrix}
    1 & -1 & 1 & 1 \\
    -1 & 1 & 1 & 1 \\
    1 & 1 & -1 & 1 \\
    1 & 1 & 1 & -1 \\
\end{matrix}\right) \otimes \left(\begin{matrix}
  1 & 1 & 1 & -1 \\
  1 & 1 & -1 & 1 \\
  1 & -1 & 1 & 1 \\
  -1 & 1 & 1 & 1 \\
\end{matrix}\right), \\
(H \otimes H) \langle e^{\imath\theta} I\rangle & \;\mapsto\; \widetilde{P_1}
\;,\\
(\sqrt{Z} \otimes \sqrt{Z}) \langle e^{\imath\theta} I\rangle & \;\mapsto\;
\widetilde{P_2}\;.
\end{align*}

A similar GAP computation shows that Spekkens' group of operations for two toy bits is not isomorphic to $\mathcal{C}(2)/U(1)$, despite the fact that both groups have $11520$ elements. One way to verify that the two groups are not isomorphic is the following: while the projective Clifford group contains no maximal subgroups of order $720$, Spekkens' group does. One such maximal subgroup is generated by the operations $(12) \otimes (23)$, $I \otimes (12)$, and $\widetilde{P_3}$ (also shown in Fig.~\ref{fig:tswap}).

As $TG(2)$ is isomorphic to the extended Clifford group--- which contains the
Clifford Group as a proper subgroup---the relaxation of STM~2 to STM~2$'$
results in a group of operations that is isomorphic to the Clifford Group of two
qubits.  We emphasize that this equivalence is a direct consequence of applying
empiricism to STM. 

Unfortunately, the relaxation of STM~2 to STM~2$'$ gives rise to a physically
unreasonable state of affairs.  For a physical model, we expect
that if an operation is valid for a given system, then it should also be valid
when we attach an ancilla to that system; the operations of $TG(2)$
violate this condition. Consider the operation  $\widetilde{H}\otimes I$: under STM~2$'$, both $\widetilde{H}$ and $I$ are valid
operations on an elementary system, yet $\widetilde{H}\otimes I$ is not a valid
operation on the composite system, as it fails to map the correlated state shown in Fig.~\ref{fig:twostates}(b) to a
valid epistemic state. In fact, the subgroups of $TG(1)$ and $TG(2)$ that preserve valid epistemic states when an ancilla is added
are simply Spekkens' original groups of operations for one and two toy bits respectively. 

However, just as positive maps serve as tests for entanglement in quantum theory, validity-preserving maps serve as tests of correlation in the toy theory, as we now explain.

Formally, let $\cA_i$ denote the set of operators acting on the Hilbert space $\cH_i$. Then a linear map $\De: \cA_1 \rightarrow \cA_2$ is \defn{positive} if it maps positive operators in $\cA_1$ to positive operators in $\cA_2$: in other words, $\rho \geq 0$ implies $\De\rho \geq 0$. On the other hand $\De$ is \defn{completely positive} if the map
\[
\De \otimes I: \cA_1 \otimes \cA_3 \rightarrow \cA_2 \otimes \cA_3
\]
is positive for every identity map $I: \cA_3 \rightarrow \cA_3$. In other words, a completely positive map takes valid density operators to valid density operators even if an ancilla is attached to the system. Also recall that an operator $\rho \in \cA_1 \otimes \cA_2$ is \defn{separable} if it
can be written in the form
\begin{equation}\label{separable}
 \varrho=\sum_{i=1}^n p_i\varrho_i\otimes\tilde{\varrho}_i,
\end{equation}
for $\varrho_i \in \cA_1$, $\tilde{\varrho}_i \in \cA_2$, and some probability distribution $\{p_i\}$. A well known result in quantum information is that positive maps can distinguish whether or not a state
is separable (Theorem 2~\cite[p.~5]{Horodecki's}):
\begin{theorem}\label{th:Horodecki}
 Let $\varrho$ act on $\mathcal{H}_1\otimes\mathcal{H}_2$. Then $\varrho$ is separable if and only if for
any positive map $\De: A_1\rightarrow A_2$, the operator $(\De \otimes I)\varrho$ is positive. 
\end{theorem}

Theorem~\ref{th:Horodecki} says is that maps that are positive but not not completely positive serve as tests for detecting
whether or not a density matrix is separable. An analogous statement can be made for validity preserving maps and correlated states in a two toy bit system. 

Define a transformation $\De$ in STM to be \defn{validity-preserving} if it maps all valid epistemic states to valid epistemic states in a toy system; all operations in $TG(1)$ and $TG(2)$ are validity-preserving. Define $\De$ to be \defn{completely validity-preserving} if $\De \otimes I$ is
validity-preserving for every $I$, where $I$ is the identity transformation on some ancilla toy system. For example, $\widetilde{H} \in TG(1)$ is validity-preserving but not completely validity-preserving. Finally, a two toy bit state is \defn{perfectly correlated} if for any acquisition of knowledge about one of the systems, the description of the other system is refined. The perfectly correlated two toy bit states are precisely the correlated pure states: no mixed states are perfectly correlated.

\begin{theorem}\label{th:toyseparability}
 Let $\sigma$ be a two toy bit epistemic state (pure or mixed).
Then $\sigma$ is perfectly correlated if and only if there exists a one toy bit validity-preserving operation $\Delta$ such that 
$(\De \otimes I)\sigma$ is an invalid two toy bit state.
\end{theorem}
\begin{proof}
First suppose $\sigma$ is a pure state. If $\sigma$ is uncorrelated, then it has the form $e_{ab}\otimes e_{cd}$, and for any $\Delta\in TG(1)$, the state
\[
(\Delta \otimes I)(e_{ab}\otimes e_{cd})=(\Delta e_{ab})\otimes e_{cd}
\]
is a valid two toy bit state. On the other hand, if $\sigma$ is correlated, then it has the form $(I\otimes P)\sigma_0$, where $\sigma_0$ is the correlated state shown in Fig.~\ref{fig:twostates}(b) and $P \in S_4$ is some permutation of the second toy bit system. In this case, the state
\[
(\widetilde{H} \otimes I)(I\otimes P)\sigma_0=(I \otimes P)(\widetilde{H} \otimes I)\sigma_0
\]
is an invalid state, as we have already seen that $(\widetilde{H}\otimes I)\sigma_0$ is invalid. 

Next suppose $\sigma$ is a mixed state. Then either $\sigma$ is uncorrelated, in which case it has the form $e_{ab}\otimes e_{1234}$, $e_{1234}\otimes e_{ab}$, or $e_{1234}\otimes e_{1234}$, or it is correlated, in which case it has the form $(e_{ab}\otimes e_{cd}+e_{mn}\otimes e_{pq})$, with $\{a,b\}$ disjoint from $\{m,n\}$ and $\{c,d\}$ disjoint from $\{p,q\}$. Any of these mixed states may be written as a sum of pure uncorrelated states. Since pure uncorrelated states remain valid under $\De \otimes I$ for any validity preserving $\De$, it follows that $(\Delta\otimes I)\sigma$ is also a valid state whenever $\sigma$ is a mixed state. Thus, invalidity of a state under a local validity-preserving map is a necessary and sufficient condition for a
bipartite epistemic state (pure or mixed) to have perfect correlation. 
\end{proof}  

In this section we introduced a possible relaxation of STM.  Motivated by
empiricism, we argued for the relaxation of STM~2, from ontic to epistemic
determinism.  We showed that this relaxation gives rise to a group of operations
that is equivalent to the projective extended Clifford Group for one and two
qubits.  However, the operations of $TG(1)$ and $TG(2)$ are physically
unreasonable as they do not represent completely validity-preserving maps. They
do, however, serve as tests for correlations in the toy model.  In
the next section we discuss these results further.

\section{\label{sec:conclusion}Discussion}

In this paper we formulated STM in an axiomatic framework and considered a
possible relaxation---STM~2$'$---in its assumptions. The motivation for
proposing STM~2$'$ is the empirical fact that in a toy universe, an observer is
restricted to knowledge of epistemic states. We
discovered that replacing STM~2 with STM~2$'$ gave rise to a group of operations
that exhibit an isomorphism with the projective extended Clifford Group of
operations (and consequently the projective Clifford group of operations) in
quantum mechanics. This characteristic is not present in STM; while $S_4$ is
isomorphic to $\mathcal{C}(1)/U(1)$, the group of operations for two toy bits in
STM is not isomorphic to $\mathcal{C}(2)/U(1)$.
However, due to the fact that operations arising from STM~2$'$ do not compose
under the tensor product---they are not completely validity-preserving---the
proposed relaxation does not give rise to a physically reasonable model.

Despite this failure, the group of operations generated by STM~2$'$ gives rise to
a very useful tool; namely, the Horodecki criterion for separability in the toy
model. The same operations that render the toy model physically unreasonable
serve as tools for detecting correlations in the toy model.  We believe that the
investigation into possible relaxations of the axioms
of STM gives rise demonstrates the power as well as the limitations of STM. 
Most significantly, we discover that no physically reasonable toy model can
arise from relaxing STM~2 to an epistemic perspective; this robustness is
an indication of the model's power.  On the other hand, we conclude that there
is at least one characteristic of quantum theory that the STM cannot capture, an
equivalence with the Clifford Group of operations. 

\begin{acknowledgments}
  The authors would like to thank Rob Spekkens for his helpful
  discussion of the two toy bit system, Nathan Babcock, Gilad Gour, and an
anonymous referee for insightful comments
and suggestions. This research has been supported by NSERC, MITACS, a CIFAR Associateship, and iCORE.
\end{acknowledgments}

\nocite{*}
 
\end{document}